\documentclass[11pt,a4paper]{article}

\usepackage[T1]{fontenc}
\usepackage[utf8]{inputenc}
\usepackage{lmodern}
\usepackage[margin=1.3in]{geometry}
\usepackage{microtype}
\usepackage{amsmath,amssymb,amsthm,mathtools}
\usepackage[numbers,sort&compress]{natbib}
\usepackage[hidelinks]{hyperref}

\newcommand{\keywords}[1]{%
  \vspace{0.75em}\noindent\textbf{Keywords: }#1\par
}

\theoremstyle{definition}
\newtheorem{definition}{Definition}
\newtheorem{postulate}{Postulate}
\theoremstyle{plain}
\newtheorem{lemma}{Lemma}
\newtheorem{proposition}{Proposition}
\theoremstyle{remark}
\newtheorem{remark}{Remark}

\title{From Kinematics to Interference: Operational Requirements for the Quantum Principle of Relativity}
\author{%
  Mikołaj Sienicki\thanks{Polish–Japanese Academy of Information Technology, ul.~Koszykowa~86, 02–008 Warsaw, Poland, European Union}%
  \quad and \quad
  Krzysztof Sienicki\thanks{Chair of Theoretical Physics of Naturally Intelligent Systems (\(\mathbb{NIS}\)\textsuperscript{\copyright}), Lipowa~2/Topolowa~19, 05–807 Podkowa Leśna, Poland, European Union; E-mail: \texttt{niskrissienicki@gmail.com}}%
}
\date{\today}

\begin{document}
\maketitle

\begin{abstract}
The ``quantum principle of relativity'' (QPR) puts forward an ambitious idea: extend special relativity with a formally
superluminal branch of Lorentz-type maps, and treat the resulting consistency constraints as hints about why quantum
theory has the structure it does \citep{DraganEkert2020}. The discussion that followed has emphasized a basic point:
writing down coordinate maps is not the same thing as providing a physical theory. In particular, quantum superposition
is not operationally defined by drawing multiple ``paths on paper'': it is defined by what happens when alternatives
\emph{recombine} in an interference loop \citep{DelSantoHorvat2022,Horodecki2024}. In parallel, careful $1{+}1$ analyses
have clarified how sign conventions and time-orientation choices enter the superluminal formulas \citep{Damski2025}.
Finally, tachyonic QFT proposals suggest a possible mathematical ``bridge'' via an enlarged (twin) Hilbert space
\citep{Paczos2024}, although this proposal remains contested (e.g.\ on commutator covariance and microcausality grounds)
\citep{Jodlowski2024Unphysical}.

The aim of this short note is organizational. We keep three layers separate: (K) kinematics (which maps exist and what
they preserve), (O) operational content (what an experiment must actually reproduce, especially closed-loop
interference), and (D/B) dynamics and bridges (how amplitudes and probabilities are generated, and how subluminal and
superluminal sectors might be linked). The goal is not ``relativity derives quantum theory,'' but a clear checklist of
what must be added for that ambition to become a well-posed programme.
\end{abstract}

\keywords{Quantum principle of relativity; superluminal Lorentz maps; operational completeness; Mach--Zehnder interference; complex amplitudes}

\section{Introduction}

It is easy to see why QPR is appealing. Special relativity is driven by a crisp invariance idea: the form of the laws
should not depend on which inertial frame you use. QPR asks whether one can push that mindset further by allowing
``superluminal'' descriptions, and then reading any resulting tensions as clues about quantum structure
\citep{DraganEkert2020}.

The main criticism is just as straightforward. A new class of coordinate maps, by itself, does not tell you what a
detector records, what counts as an outcome, or why an interferometer produces fringes. Those are operational questions,
and they are exactly where quantum theory earns its keep. In particular, ``superposition'' is not a slogan about having
several pictures of one trajectory; operationally it is about \emph{recombination} and \emph{phase dependence} in closed
loops \citep{DelSantoHorvat2022}. This is why several comments insist that any claimed ``derivation'' must be explicit
about which extra assumptions are being introduced \citep{Horodecki2024}.

This paper does not try to settle the debate, and it does not defend the strong claim that relativity alone implies
quantum theory. Instead, it tries to do three practical things:
\begin{itemize}
\item keep the $1{+}1$ kinematics tidy (including sign and orientation choices that are easy to gloss over);
\item state operational requirements in a way that cannot be evaded by coordinate relabeling;
\item spell out what a genuine ``bridge'' would have to provide if one wants more than a suggestive narrative.
\end{itemize}

The overall message is constructive. One can write down a coherent programme (call it ``QPR+'')---but it needs deliberate
additions. Kinematics alone will not produce loop interference; and a serious superluminal sector requires an explicit
dynamical framework that outputs observable statistics.

\section{Programme statement: kinematics, operations, dynamics}

Writing Lorentz-like formulas for $|V|>c$ is the easy part. The hard questions are different:
\begin{itemize}
\item Do such maps correspond to anything like physical ``frame changes'' in $1{+}3$?
\item Even if we accept them as auxiliary redescriptions, do they force genuinely quantum phenomena, such as loop
interference, rather than merely suggestive stories?
\end{itemize}
To keep these issues from being mixed together, we separate three layers.

\begin{definition}[Kinematic layer]
A kinematic layer specifies a class of affine linear maps between coordinate systems (including translations) used to
relate descriptions of events and worldlines, together with their algebraic properties and stated domain of use. A
kinematic map need not be a physical symmetry of Minkowski spacetime.
\end{definition}

\begin{definition}[Operational layer]
An operational layer specifies preparations, measurements, and outcome statistics, and states what is required to
count as a successful reproduction of an effect (for example, closed-loop interference).
\end{definition}

\begin{definition}[Admissible redescriptions (for QPR+; null-structure level)]
A redescription is called admissible if it is an affine map between coordinate systems on Minkowski space whose linear
part is invertible, preserves inertial straight lines, and maps null lines to null lines (equivalently: it preserves
the null cone as a set). In spacetime dimension $N>2$ (in particular $1{+}3$), standard results imply that such linear
null-cone-preserving maps are conformal Lorentz (proportional to Lorentz transformations) and therefore do not
interchange timelike and spacelike directions; see, e.g., Lemma~4.3 of \citet{BergqvistSenovilla2001}. By contrast, in
$1{+}1$ there exist additional null-cone-preserving maps with ``interval sign-flip'' behaviour. In this paper we treat
those $1{+}1$ maps as auxiliary redescriptions for kinematic analysis, not as physical inertial-frame symmetries.
\end{definition}

\begin{definition}[Physical causal past]
$J^{-}(e)$ denotes the physical causal past defined by the standard Minkowski lightcone (the one operationally fixed by
light signals and preserved by Lorentz isometries). When we use non-isometric auxiliary redescriptions, we do \emph{not}
require them to preserve $J^\pm$; they are introduced only to compare explanatory forms.
\end{definition}

\begin{definition}[Dynamical/bridge layer]
A dynamical/bridge layer provides an evolution rule or QFT structure that generates measurable statistics, and, when
needed, a precise mechanism connecting subluminal and superluminal sectors in one coherent mathematical representation.
\end{definition}

\begin{remark}
A concise way to summarize the operational critique of \citet{DelSantoHorvat2022} is this: QPR moves from (K) to (O)/(D)
too quickly. A kinematic extension does not automatically deliver recombination interference.
\end{remark}

\section{Superluminal Lorentz extensions: what kinematics gives (and what it does not)}

\paragraph{$1{+}1$ formulas and the extra sign choice.}
In $1{+}1$ spacetime, standard Lorentz transformations between inertial coordinates $(t,x)$ and $(t',x')$ are
\begin{equation}
x'=\gamma(x-Vt),\qquad t'=\gamma\left(t-\frac{V}{c^{2}}x\right),\qquad
\gamma=\frac{1}{\sqrt{1-V^{2}/c^{2}}},\quad |V|<c.
\label{eq:lorentz}
\end{equation}
Under certain assumption sets (and in some ``algebraic extension'' discussions), one also writes a second branch,
formally defined for $|V|>c$ (see also \citealt{Damski2025}):
\begin{equation}
x'=\eta(V)\,\widetilde{\gamma}(x-Vt),\qquad
t'=\eta(V)\,\widetilde{\gamma}\left(t-\frac{V}{c^{2}}x\right),\qquad
\widetilde{\gamma}=\frac{1}{\sqrt{V^{2}/c^{2}-1}},\quad |V|>c,
\label{eq:superlum}
\end{equation}
where $\eta(V)\in\{\pm 1\}$. Unlike the subluminal case, there is no $V\to 0$ limit that fixes $\eta$ by continuity. So
even in $1{+}1$, the formulas themselves do not settle their physical interpretation.

\paragraph{Interval sign flip (so: not a Lorentz symmetry).}
Although \eqref{eq:superlum} looks familiar, it does not preserve the Minkowski quadratic form. One finds
\begin{equation}
x'^2 - c^2 t'^2 \;=\; -\bigl(x^2 - c^2 t^2\bigr)\qquad(\eta^2=1),
\label{eq:interval_flip}
\end{equation}
so timelike and spacelike separations are exchanged. In other words, \eqref{eq:superlum} is an anti-isometry in
$1{+}1$. This is exactly why it is not, on its own, a standard SR ``change of inertial frame.''

\paragraph{Why $1{+}3$ is different.}
In $1{+}3$, QPR’s key warning is that if one tries to treat subluminal Lorentz transformations and a superluminal branch
on equal ``relativity'' footing, one is pushed toward non-isometric maps (e.g.\ anisotropic scalings) as putative
symmetries, which is not empirically acceptable \citep{DraganEkert2020}. So in $1{+}3$, superluminal maps should not be
treated as physical observer equivalences. At most, they can serve as auxiliary redescriptions, typically with altered
roles for time and space.

\paragraph{Role of Damski.}
Damski’s contribution is clean and kinematic: it clarifies parametrizations and sign/time-orientation conventions for
$1{+}1$ superluminal maps \citep{Damski2025}. It does not (and does not claim to) produce operational predictions such as
loop interference, nor a probability rule.

\section{QPR, determinism, and what ``indeterminism'' means in this setting}

\paragraph{The determinism assumption should be stated out loud.}
QPR’s ``indeterminism'' argument relies on a particular notion of determinism: roughly, that certain event parameters are
fixed by data local to the emitter’s own past worldline \citep{DraganEkert2020}. Critics have stressed that this is a
stronger requirement than textbook SR and should be stated explicitly \citep{DelSantoHorvat2022,Horodecki2024}.

\begin{definition}[Past-worldline data]
Let $S$ be an emitting system with timelike worldline $\Gamma_S$, equipped with an intrinsic time order (e.g.\ proper
time $\tau$). For an event $e\in\Gamma_S$, define
\begin{equation}
\Gamma_S^{-}(e):=\{e'\in\Gamma_S:\tau(e')<\tau(e)\}.
\label{eq:past_worldline}
\end{equation}
The past-worldline data at $e$ is $\mathcal{I}^{-}_{S}(e)$, the restriction of the system’s classical state variables to
$\Gamma_S^{-}(e)$.
\end{definition}

\begin{definition}[Local determinism (QPR sense)]
A process involving an event $e$ on $\Gamma_S$ is locally deterministic if there exists a function $F$ such that an
operational parameter of the event (for example, its proper-time location $\tau(e)$, or a binary occurrence indicator
$\chi(e)\in\{0,1\}$) is determined by $\mathcal{I}^{-}_{S}(e)$:
\begin{equation}
\begin{aligned}
\tau(e) &= F\!\left(\mathcal{I}^{-}_{S}(e)\right)
\quad\text{or}\quad
\chi(e)=F\!\left(\mathcal{I}^{-}_{S}(e)\right).
\end{aligned}
\label{eq:local_determinism}
\end{equation}
\end{definition}

\paragraph{QPR as a programme-level requirement.}
Because superluminal maps are not physical symmetries in $1{+}3$, treating them ``like'' frame changes is ultimately a
normative move: it states what the programme \emph{wants} to be true, rather than what SR already guarantees.

\begin{postulate}[QPR (invariance of local-deterministic admissibility)]
For any process $P$ in the theory’s domain, the statement ``$P$ admits a locally deterministic description'' has the same
truth value under all admissible redescriptions within the regime where the programme claims to apply.
\end{postulate}

\begin{proposition}[Conditional tension result (informal)]
Assume (i) physically meaningful processes exist that are described as superluminal exchanges, and (ii) local determinism
is taken in the ``past-worldline sufficiency'' sense without using information outside the physical causal past $J^{-}(e)$
of the emission event. Then different admissible descriptions can disagree on whether a locally deterministic causal
explanation exists that uses only $\mathcal{I}^{-}_{S}(e)$ and data supported in $J^{-}(e)$. Under QPR, one must then
either give up this form of local determinism (for that class of processes) or introduce preferred structures.
\end{proposition}

\begin{remark}
This is a conditional conclusion tied to a specific determinism notion. It is not the claim that SR alone implies
quantum theory. This is consistent with the more modest tone of the reply \citep{DraganEkertReply2023}.
\end{remark}

\paragraph{Why ``superposition from coordinate maps'' overreaches.}
A coordinate change cannot literally turn a single continuous worldline into a branching ``Y'' curve. So if one wants a
``many paths'' picture, that is extra structure; it does not follow from relabeling events alone \citep{DelSantoHorvat2022}.

\begin{definition}[Classical trajectory]
A classical trajectory is the image of a continuous injective map $\gamma:[0,1]\to\mathbb{R}^{1+1}$ (or
$\mathbb{R}^{1+3}$), up to reparametrization.
\end{definition}

\begin{lemma}[No branching under homeomorphisms]
Let $h$ be a homeomorphism of spacetime. If $\gamma([0,1])$ is homeomorphic to an interval, then $h(\gamma([0,1]))$ is
also homeomorphic to an interval and cannot be homeomorphic to a ``Y'' graph.
\end{lemma}

\begin{proof}[Sketch]
Removing an interior point of an interval yields two connected components; removing the branching point of a Y-graph
yields three. Homeomorphisms preserve this invariant.
\end{proof}

\section{Operational completion: interference is the real test (and it forces amplitudes)}

\paragraph{Interference loops are the benchmark.}
If one wants to talk about superposition in an operational sense, the key benchmark is closed-loop interference:
alternatives must be able to recombine, and the output statistics must depend continuously on a tunable phase (as in a
Mach--Zehnder interferometer) \citep{DelSantoHorvat2022}. We treat this as non-optional: any ``quantum-like'' account has
to get this right.

\begin{postulate}[O1: Detector-model neutrality]
Any ``quantum-like'' claim should be robust across reasonable detector models, including non-demolishing and
absorb-and-reemit devices \citep{DelSantoHorvat2022}.
\end{postulate}

\begin{postulate}[O2: Loop completeness]
A framework that claims to recover superposition must reproduce closed-loop interference (e.g.\ Mach--Zehnder
recombination) with a continuous dependence of output statistics on a tunable relative phase or delay, in setups where
the alternatives are operationally indistinguishable at recombination \citep{DelSantoHorvat2022}.
\end{postulate}

\begin{postulate}[O3: Statistics invariance under physical symmetries (optional strengthening omitted)]
If probabilistic predictions are made for an operationally specified experiment in one physical inertial frame, then any
other physical inertial frame related by a Lorentz isometry must yield the same measurable statistics.
\end{postulate}

\begin{definition}[Probability rule (minimal)]
A probability rule is a functional $P(\cdot)$ mapping amplitudes (or equivalence classes thereof) to outcome weights
interpretable as probabilities, subject at minimum to: (i) invariance under global phase $P(uA)=P(A)$ for the action
$A\mapsto u(\phi)A$ of Postulate~A4, and (ii) operational normalization and additivity for mutually exclusive outcomes.
To match O2 one also assumes that there exist setups in which $P$ depends continuously and nontrivially on a relative
phase parameter for operationally indistinguishable alternatives. No specific choice such as $P(A)=|A|^{2}$ is assumed.
\end{definition}

\begin{remark}
O3 enforces invariance under physical (Lorentz) frame changes. It is not, by itself, a derivation of the Born rule. This
matches the incompleteness concerns emphasized in \citep{Horodecki2024}.
\end{remark}

\paragraph{Why kinematics plus classical mixtures cannot pass O2.}
Once O2 is taken seriously, a story that only uses coordinate maps, worldlines, and classical mixture rules cannot do the
job.

\begin{definition}[Kinematics-only model (event structure + classical mixtures)]
A kinematics-only model is one in which (i) the ontology consists of classical events and worldlines, (ii) admissible
transformations are coordinate maps, and (iii) experiment probabilities are computed by a classical mixture rule over
alternatives (adding nonnegative weights), without introducing a separate phase-like quantity whose relative value
changes recombination statistics for operationally indistinguishable alternatives.
\end{definition}

\begin{proposition}[No-go for kinematics-only loop completeness]
\label{prop:nogo_kin_only}
No kinematics-only model (in the above sense) can satisfy O2 for a Mach--Zehnder-type interferometer. To reproduce
continuous interference fringes for operationally indistinguishable alternatives, one must leave the classical mixture
rule and introduce a phase-sensitive composition rule (amplitudes) at recombination.
\end{proposition}

\begin{proof}[Proof sketch]
In a Mach--Zehnder interferometer, changing a relative phase $\phi$ in one arm can continuously move probability weight
between the output ports, including destructive cancellation at one port in an idealized setup. A classical mixture over
alternatives has the schematic form
\begin{equation}
P = \sum_k w_k\,P_k,\qquad w_k\ge 0,\ \sum_k w_k=1,
\label{eq:classical_mixture}
\end{equation}
and it cannot produce phase-controlled cancellation between indistinguishable alternatives.

One quick way to see the limitation is to block one arm. Unconditioned on loss, a 50/50 first beam splitter sends half
the intensity (or probability weight) into the blocked arm and removes it; the remaining half reaches the second beam
splitter and splits again, giving $P(D_0)=P(D_1)=1/4$ and a loss channel $P(\text{absorb})=1/2$. Conditioned on detection
at the outputs, one gets $P(D_0\mid \text{detected})=P(D_1\mid \text{detected})=1/2$, still with no dependence on $\phi$.
Either way, the mixture description stays phase-independent. To get cancellation as $\phi$ varies, one needs a signed or
complex composition law, i.e.\ amplitudes that depend on relative phase \citep{DelSantoHorvat2022}.
\end{proof}

\begin{remark}
This proposition forces a phase-sensitive (wave-like) composition structure, but not yet uniquely a quantum one.
Classical wave optics also uses phase-sensitive addition. ``Quantum'' enters once one supplies a probability map
(Born-type or equivalent) and a consistent account of measurement/which-way information in line with O1--O3.
\end{remark}

\paragraph{Minimal amplitude calculus.}
QPR naturally points toward amplitudes as a convenient way to build a covariant, phase-sensitive calculus, while also
noting that kinematics alone does not fix the phase scale or the dynamics \citep{DraganEkert2020,DraganEkertReply2023}.
Here we simply make explicit the minimal structural assumptions needed once O2 is treated as a hard requirement.

\begin{remark}
Postulates A1--A2 are assumptions about how alternatives and concatenations compose. They are not consequences of
kinematics alone.
\end{remark}

\begin{postulate}[A1: Sum rule for indistinguishable alternatives]
If $\gamma_1$ and $\gamma_2$ are alternatives that are not distinguished (no which-way information) and can later
recombine, then there exists an assignment $A(\gamma)$ such that
\begin{equation}
A(\gamma_1 \ \text{or}\ \gamma_2)=A(\gamma_1)+A(\gamma_2).
\label{eq:sum_rule}
\end{equation}
\end{postulate}

\begin{postulate}[A2: Product rule for concatenation]
For sequential concatenation $\gamma=\gamma_{AB}\circ \gamma_{BC}$,
\begin{equation}
A(\gamma)=A(\gamma_{AB})\,A(\gamma_{BC}).
\label{eq:product_rule}
\end{equation}
\end{postulate}

\begin{postulate}[A3: Nontrivial interference exists]
There exist alternatives for which observed recombination statistics are not reducible to a classical mixture;
operationally, this is exactly what O2 demands in suitable setups.
\end{postulate}

\begin{postulate}[A4: Continuous phase symmetry]
There exists a continuous one-parameter multiplicative group action $u(\phi)$ on amplitudes such that global
$A\mapsto u(\phi)A$ leaves probabilities invariant, while relative phases can change recombination statistics.
\end{postulate}

\begin{proposition}[Minimal carrier (two real dimensions; complex numbers as a standard realization)]
Under A1--A4 and mild regularity assumptions (continuity; non-degeneracy of the phase action), a one-dimensional real
scalar carrier is insufficient once one asks for a continuous phase action as in A4. The minimal workable carrier is
two-dimensional over $\mathbb{R}$ with a unit-circle subgroup acting multiplicatively. With standard algebraic
coherence conditions for composition (associativity, distributivity, identity), this structure is naturally realized as
$\mathbb{C}$ with $u(\phi)=e^{i\phi}$ (equivalently, a 2D real representation of $U(1)$).
\end{proposition}

\begin{remark}
This explains why complex amplitudes (or an equivalent 2D real phase structure) are the simplest upgrade once one insists
on loop interference and composition. It does not derive the Born rule. The constant $\hbar$ appears only when mapping
dynamics to phase; QPR is right that kinematics alone cannot fix it \citep{DraganEkert2020}.
\end{remark}

\section{Bridge and outlook: the ``no footbridge'' objection}

Horodecki’s diagnosis is that without a concrete structure connecting subluminal and superluminal sectors, QPR remains
incomplete \citep{Horodecki2024}. A programme can respond in two ways: either narrow the claim (QPR as a heuristic) or
propose an explicit bridge and then test it.

A candidate bridge is the twin-Hilbert-space tachyon QFT suggested by Paczos et al.\ \citep{Paczos2024}. The intended
role is structural: provide a representation in which Lorentz covariance, commutators, and a vacuum state can be handled
in an enlarged state space. Note that ``tachyonic'' in QFT often refers to instabilities (negative $m^2$); here ``tachyon''
is used in the superluminal-sector sense of \citealt{Paczos2024}. Also note that the proposal concerns covariance under
the standard Lorentz group; connecting that to QPR’s auxiliary superluminal redescriptions would require an additional,
explicit identification step.

\begin{postulate}[B1: Twin-space bridge (candidate target properties)]
A candidate bridge aims to realize an enlarged state space $\mathcal{H}_{\mathrm{twin}}=\mathcal{F}\otimes\mathcal{F}^\star$
together with a unitary representation $U$ of the standard Lorentz group (or a specified subgroup) acting on
$\mathcal{H}_{\mathrm{twin}}$ such that:
(i) commutation relations are preserved under $U(\Lambda)$;
(ii) an invariant vacuum-like state exists on $\mathcal{H}_{\mathrm{twin}}$; and
(iii) amplitudes for processes are computed as a $c$-number functional (for example, trace-type) on a suitable subspace,
yielding an operationally meaningful amplitude calculus.
\end{postulate}

\begin{remark}[Status and limitations of B1]
Paczos et al.\ propose a construction intended to meet B1 \citep{Paczos2024}, but it has been challenged. For example,
\citet{Jodlowski2024Unphysical} argues that (as formulated) the commutator is not Lorentz invariant and that microcausality
fails except in a limiting case. Also, B1 by itself does not guarantee operational adequacy: one must still show
explicitly how measurement models and interferometric loops are represented so that O2 (loop completeness) and O3
(statistics invariance) hold in the proposed formalism.
\end{remark}

\paragraph{One-line summary.}
Kinematics can point to a conditional tension with a strong form of local determinism; operational adequacy forces a
phase-sensitive composition rule (amplitudes); and any serious ``derivation'' programme needs an explicit dynamical
bridge---not just coordinate maps.

\bibliographystyle{unsrtnat}
\bibliography{refs}

\begin{thebibliography}{8}
\providecommand{\natexlab}[1]{#1}
\providecommand{\url}[1]{\texttt{#1}}
\expandafter\ifx\csname urlstyle\endcsname\relax
  \providecommand{\doi}[1]{doi: #1}\else
  \providecommand{\doi}{doi: \begingroup \urlstyle{rm}\Url}\fi

\bibitem[Dragan and Ekert(2020)]{DraganEkert2020}
Andrzej Dragan and Artur Ekert.
\newblock Quantum principle of relativity.
\newblock \emph{New Journal of Physics}, 22:\penalty0 033038, 2020.
\newblock \doi{10.1088/1367-2630/ab76f7}.
\newblock URL \url{https://iopscience.iop.org/article/10.1088/1367-2630/ab76f7/pdf}.

\bibitem[Del~Santo and Horvat(2022)]{DelSantoHorvat2022}
Flavio Del~Santo and Sebastian Horvat.
\newblock Comment on {``Quantum principle of relativity''}.
\newblock \emph{New Journal of Physics}, 24\penalty0 (12):\penalty0 128001, 2022.
\newblock \doi{10.1088/1367-2630/acae3b}.
\newblock URL \url{https://iopscience.iop.org/article/10.1088/1367-2630/acae3b/pdf}.

\bibitem[Horodecki(2024)]{Horodecki2024}
Ryszard Horodecki.
\newblock Comment on {``Quantum principle of relativity''}.
\newblock 2024.
\newblock URL \url{https://arxiv.org/abs/2301.07802}.
\newblock arXiv:2301.07802v3 [quant-ph], 19 Jan 2024.

\bibitem[Damski(2025)]{Damski2025}
B.~S. Damski.
\newblock Lorentz transformations in 1+1 dimensional spacetime: Mainly the superluminal case.
\newblock \emph{Acta Physica Polonica A}, 148\penalty0 (1):\penalty0 22, 2025.
\newblock \doi{10.12693/APhysPolA.148.22}.
\newblock URL \url{https://www.if.pw.edu.pl/~appa/appa148z1.html}.

\bibitem[Paczos et~al.(2024)Paczos, D\k{e}bski, Cedrowski, Charzy\'{n}ski, Turzy\'{n}ski, Ekert, and Dragan]{Paczos2024}
Jerzy Paczos, Kacper D\k{e}bski, Szymon Cedrowski, Szymon Charzy\'{n}ski, Krzysztof Turzy\'{n}ski, Artur Ekert, and Andrzej Dragan.
\newblock Covariant quantum field theory of tachyons.
\newblock 2024.
\newblock \doi{10.48550/arXiv.2308.00450}.
\newblock URL \url{https://arxiv.org/abs/2308.00450}.
\newblock arXiv:2308.00450v2 [quant-ph], 25 Jun 2024.

\bibitem[Jod{\l}owski(2024)]{Jodlowski2024Unphysical}
Krzysztof Jod{\l}owski.
\newblock Covariant quantum field theory of tachyons is unphysical.
\newblock 2024.
\newblock \doi{10.48550/arXiv.2406.14225}.
\newblock URL \url{https://arxiv.org/abs/2406.14225}.
\newblock arXiv:2406.14225v5 [quant-ph], 9 Dec 2024.

\bibitem[Bergqvist and Senovilla(2001)]{BergqvistSenovilla2001}
G{\"o}ran Bergqvist and Jos{\'e} M.~M. Senovilla.
\newblock Null cone preserving maps, causal tensors and algebraic {R}ainich theory.
\newblock \emph{Classical and Quantum Gravity}, 18\penalty0 (23):\penalty0 5299--5325, 2001.
\newblock \doi{10.1088/0264-9381/18/23/323}.
\newblock URL \url{https://arxiv.org/abs/gr-qc/0104090}.

\bibitem[Dragan and Ekert(2023)]{DraganEkertReply2023}
Andrzej Dragan and Artur Ekert.
\newblock Reply to the comment on {``Quantum principle of relativity''}.
\newblock \emph{New Journal of Physics}, 25\penalty0 (12):\penalty0 128002, 2023.
\newblock \doi{10.1088/1367-2630/ad100e}.
\newblock URL \url{https://iopscience.iop.org/article/10.1088/1367-2630/ad100e/pdf}.

\end{thebibliography}

\end{document}